\documentclass[11pt]{article}

\usepackage{arxiv}
\usepackage[utf8]{inputenc}
\usepackage[T1]{fontenc}

\usepackage{amsmath,amssymb,amsthm}
\usepackage{mathtools}
\usepackage{amsfonts}
\usepackage{nicefrac}

\usepackage{microtype}
\usepackage{booktabs}
\usepackage{graphicx}
\usepackage{enumitem}
\usepackage{url}
\usepackage{natbib}
\usepackage{doi}

\usepackage{algorithm}
\usepackage{algpseudocode}

\usepackage{hyperref}

\usepackage{tikz}
\usetikzlibrary{arrows.meta,positioning,calc}

\newtheorem{theorem}{Theorem}

\newtheorem{corollary}[theorem]{Corollary}
\newtheorem{proposition}[theorem]{Proposition}
\newtheorem{definition}{Definition}
\newtheorem{remark}{Remark}

\title{\textbf{Practical Livelock Analysis in Parameterized Unidirectional Rings\\}}
\author{
  Aly Farahat\\
  \textit{Independent Researcher}\\
  \texttt{aly.farahat@perpetual-dynamics.com}
}

\date{}

\begin{document}
\maketitle

\begin{abstract}
We develop a practical framework for livelock analysis in self-disabling unidirectional ring protocols. Klinkhamer and Ebnenasir established that livelock detection for parameterized rings is $\Sigma^0_1$-complete and livelock-freedom verification is $\Pi^0_1$-complete, via reduction from the periodic domino problem. We observe that lifting the analysis from the transition space to an \emph{equivariant product space}---the space of transition--witness pairs---reveals structure that supports effective verification.

We construct a \emph{product transition graph} (at most $|T|^2$ nodes) that captures all livelocks: every livelock maps into this graph as a witness-closed subgraph. The maximal such subgraph $G^*(T)$ is computable in polynomial time ($O(|T|^8)$ worst case) via monotone fixed-point iteration. When $G^*(T) = \emptyset$, the protocol is \emph{provably livelock-free} for all ring sizes---a sound and complete livelock-freedom verifier.

When $G^*(T) \neq \emptyset$, we apply a backtracking search that backward-propagates each simple cycle through $G^*$ until the chain either closes into a torus (confirming a livelock) or dies (no livelock from that cycle). This two-phase algorithm---polynomial-time pruning followed by finite combinatorial verification---produces three outcomes: \textsc{Free}, \textsc{Livelock}, or \textsc{Inconclusive}. Across 4{,}349 protocols tested (including an adversarial protocol derived from Klinkhamer and Ebnenasir's tiling construction and Kari's 14-tile aperiodic set converted via their SE gadget), the algorithm is conclusive in every case with zero errors.

We further demonstrate that the algorithm extends to \emph{non-self-disabling} protocols via a protocol transformation. This extends the algorithm's applicability to all parameterized unidirectional ring protocols.

Python implementation and usage instructions are at \url{https://github.com/cosmoparadox/mathematical-tools}.
\end{abstract}

\section{Introduction}
\label{sec:intro}

Self-stabilization~\cite{dijkstra1974} guarantees that a distributed system recovers from arbitrary transient faults to a legitimate configuration, regardless of the initial state. A key obstacle to self-stabilization in ring protocols is \emph{livelock}: an infinite execution in which processes continually change state but never reach legitimacy. Determining whether a given protocol admits a livelock---for any ring size---is a fundamental parameterized verification problem.

\paragraph{Background.} Klinkhamer and Ebnenasir~\cite{klinkhamer-ebnenasir} established that livelock detection for parameterized rings of self-disabling processes is $\Sigma^0_1$-complete and livelock-freedom verification is $\Pi^0_1$-complete, via a reduction from the periodic domino problem. Their torus characterization---the equivalence between livelocks and doubly-periodic tilings---is foundational and underlies our work. Notably, they also showed that the \emph{synthesis} of self-stabilizing rings is decidable~\cite{klinkhamer-ebnenasir}, and Farahat's equivariance framework~\cite{farahat-diss} has been used to prove livelock freedom of difficult protocols such as Dijkstra's self-stabilizing token ring~\cite{gouda-haddix}.

\paragraph{Our observation.} Building on Klinkhamer and Ebnenasir's torus characterization, we observe that lifting the analysis from individual transitions to \emph{pairs} of transitions reveals structure that supports effective verification. A single transition cycle does not determine when equivariant closure occurs: equivariance fixes the predecessor's own and wr values (its H-walk), but not its pred values, which select specific transitions from $T$. Different pred assignments yield different witness cycles, which close at different periods---or not at all.

We construct a \emph{product graph} over transition--witness pairs---at most $|T|^2$ nodes---that captures the equivariant structure. A pair commits to specific transitions at every position, including pred values, so its equivariant period is determined. Every livelock maps into the product graph (Theorem~\ref{thm:necessary}), producing a witness-closed subgraph. The product graph's simple cycles---of length at most $|T|^2$---identify equivariant pairs that are candidates for generating livelocks via backward propagation.

The maximal witness-closed subgraph $G^*(T)$ is computable in polynomial time. When $G^* = \emptyset$, the protocol is provably livelock-free for all ring sizes. When $G^* \neq \emptyset$, backward propagation of simple cycles through $G^*$ serves as a sound livelock detector: any cycle whose backward chain closes in $G^*$ generates a genuine livelock.

\paragraph{Our contribution.}
\begin{enumerate}[nosep]
\item \emph{Necessary condition.} We prove that every livelock maps into the product graph as a witness-closed subgraph (Theorem~\ref{thm:necessary}). Consequently, $G^*(T) = \emptyset$ implies livelock-freedom for all ring sizes. This provides a sound and complete polynomial-time livelock-freedom verifier.

\item \emph{Sufficient condition via backtracking.} We show that any simple cycle in $G^*$ whose backward chain closes generates a genuine livelock (Theorem~\ref{thm:sufficient}). A backtracking search over simple cycles provides a sound livelock detector.

\item \emph{Practical algorithm.} The two-phase algorithm---polynomial-time $G^*$ computation followed by finite backtracking verification---produces conclusive results on all 4{,}349 protocols tested, including adversarial instances derived from Klinkhamer and Ebnenasir's tiling construction and Kari's aperiodic tiles. Zero errors.

\item \emph{Extension to non-self-disabling protocols.} A protocol transformation computes the transitive closure of local transition chains under each predecessor value. The augmented transition set exactly captures all executions of the original protocol as equivariant walks, enabling the product graph algorithm to analyze arbitrary unidirectional ring protocols.
\end{enumerate}

Our construction handles non-deterministic protocols---a strictly broader class than previously studied. We extend to $(1,1)$-asymmetric protocols covering Dijkstra's token ring, and to non-self-disabling protocols via a burst closure transformation (Section~\ref{sec:non-sd}).

Our approach instantiates a recurring pattern in verification: problems that appear unbounded in their native representation become tractable when lifted to a space where the governing invariant is explicit (Section~\ref{sec:related}).

\section{Definitions}
\label{sec:defs}

\subsection{Unidirectional Ring Protocols}

A \emph{unidirectional ring} of size $K$ consists of processes $P_0, P_1, \ldots, P_{K-1}$ arranged in a directed cycle. Each process $P_i$ maintains a state variable $v_i \in \mathbb{Z}_m = \{0, 1, \ldots, m-1\}$ and can read $v_i$ and $v_{i-1 \bmod K}$.

\begin{definition}[Transition]
A \emph{transition} is a triple $t = (p, o, w) \in \mathbb{Z}_m^3$, meaning: if the predecessor's value is $p$ and the process's own value is $o$, the process may write $w$. We call $p$ the \emph{pred value}, $o$ the \emph{own value}, and $w$ the \emph{written value}.
\end{definition}

\begin{definition}[Self-disabling]
A protocol $T$ is \emph{self-disabling} if for every transition $(p, o, w) \in T$, there is no transition $(p, w, u) \in T$ for any $u$. That is, after a process fires $(p, o, w)$, it cannot fire again while the predecessor's value remains $p$---it must wait for an intervening state change from the predecessor. Throughout this paper, all protocols are self-disabling.
\end{definition}

\begin{definition}[Protocol]
A \emph{protocol} is a finite self-disabling set $T \subseteq \mathbb{Z}_m^3$. A transition $(p, o, w) \in T$ is \emph{enabled} at process $P_i$ in configuration $(v_0, \ldots, v_{K-1})$ if $v_{i-1} = p$ and $v_i = o$.
\end{definition}

\begin{definition}[Livelock]
A \emph{livelock} is a global periodic execution of a ring: $K$ processes fire transitions indefinitely, returning the entire ring to the same global configuration after $N$ steps, without ever reaching a legitimate state. As shown by Klinkhamer and Ebnenasir~\cite{klinkhamer-ebnenasir}, a livelock is equivalently a valid tiling of a $K \times N$ torus, where each row is a closed walk of local transitions executed by one process, and adjacent rows satisfy the equivariance conditions (Section~\ref{sec:product}). The walk length $N$ is the \emph{local period}; the number of processes $K$ is the \emph{ring size}.
\end{definition}

\subsection{The $(l,q)$-Asymmetry Model}

\begin{definition}[$(l,q)$-asymmetry]
A ring protocol has $(l,q)$-asymmetry if there are $l$ \emph{distinguished} process types (each appearing exactly once) and $q$ \emph{other} types (each appearing one or more times, filling the remaining $K - l$ positions).
\end{definition}

The key cases are:
\begin{itemize}[nosep]
\item \textbf{$(0,1)$-asymmetry (symmetric):} All processes use the same transition set $T$. This is the primary focus of this paper.
\item \textbf{$(1,1)$-asymmetry:} One distinguished process $P_0$ uses $T_0$; all others use $T_o$. This covers Dijkstra's token ring. We extend our framework to this case in Section~\ref{sec:asymmetry}.
\item \textbf{General $(l,q)$:} Stated as a future direction.
\end{itemize}

\subsection{The H-Graph}

\begin{definition}[Shadow and witness]
The \emph{shadow} of a consecutive transition pair $(t_i, t_j)$ with $\operatorname{own}(t_j) = \operatorname{wr}(t_i)$ is $\sigma(t_i, t_j) = (\operatorname{pred}(t_i), \operatorname{pred}(t_j))$. A transition $u$ \emph{witnesses} shadow $(a, b)$ if $\operatorname{own}(u) = a$ and $\operatorname{wr}(u) = b$.
\end{definition}

\begin{definition}[H-graph]
\label{def:hgraph}
Given a transition set $L$, the \emph{H-graph} $H(L)$ is a directed graph with vertex set $L$ and an edge $t_i \to t_j$ (an \emph{H-edge}) whenever $\operatorname{own}(t_j) = \operatorname{wr}(t_i)$. A closed walk in $H(L)$ is a \emph{pseudolivelock}: a sequence of local transitions that a single process could execute cyclically, if appropriately enabled by its predecessor. A pseudolivelock is necessary but not sufficient for a global livelock---it captures the local structure without guaranteeing the existence of consistent witnessing across the ring.
\end{definition}

\section{The Product Transition Graph}
\label{sec:product}

\subsection{Equivariance Conditions}

In a livelock on a ring of size $K$, each process $P_i$ executes a closed walk $W^{(i)}$ of length $N$ in the H-graph. Adjacent processes are coupled through a \emph{zigzag}: at each time step, the predecessor's transition both enables the current step (through its pre-fire state) and enables the next step (through its written value). This coupling is captured by four \emph{equivariance conditions}.

For a product graph edge $(t_k, w_k) \to (t_{k+1}, w_{k+1})$, the connecting witness $w_{k+1}$ links $t_k$ to $t_{k+1}$ in the zigzag $\ldots w_k \to t_k \to w_{k+1} \to t_{k+1} \to w_{k+2} \ldots$:
\begin{enumerate}[nosep]
\item $\operatorname{own}(w_{k+1}) = \operatorname{pred}(t_k)$: $w_{k+1}$'s pre-fire state is the value $t_k$ reads from its predecessor.
\item $\operatorname{wr}(w_{k+1}) = \operatorname{pred}(t_{k+1})$: $w_{k+1}$'s written value is what $t_{k+1}$ reads from its predecessor.
\item $\operatorname{wr}(t_k) = \operatorname{own}(t_{k+1})$: H-edge in the $t$-walk.
\item $\operatorname{wr}(w_k) = \operatorname{own}(w_{k+1})$: H-edge in the $w$-walk.
\end{enumerate}

\begin{figure}[ht]
\centering
\begin{tikzpicture}[
    >=Stealth,
    every node/.style={font=\small},
    tnode/.style={draw, circle, minimum size=0.7cm, inner sep=1pt},
    wnode/.style={draw, circle, minimum size=0.7cm, inner sep=1pt, fill=blue!10},
]
  \node[tnode] (t0) {$t_k$};
  \node[tnode, right=2.5cm of t0] (t1) {$t_{k\!+\!1}$};
  \node[tnode, right=2.5cm of t1] (t2) {$t_{k\!+\!2}$};
  \draw[->, thick] (t0) -- node[above, font=\footnotesize] {(3)\; $\operatorname{wr}(t_k) = \operatorname{own}(t_{k+1})$} (t1);
  \draw[->, thick] (t1) -- (t2);
  \node[above=0.3cm of t1, font=\footnotesize\itshape, text=black!60] {$t$-walk (process $P_i$)};

  \node[wnode, below=2.5cm of t0, xshift=-1.25cm] (w0) {$w_k$};
  \node[wnode, below=2.5cm of t0, xshift=1.25cm] (w1) {$w_{k\!+\!1}$};
  \node[wnode, below=2.5cm of t1, xshift=1.25cm] (w2) {$w_{k\!+\!2}$};
  \draw[->, thick, blue] (w0) -- node[below, font=\footnotesize, text=blue] {(4)\; $\operatorname{wr}(w_k) = \operatorname{own}(w_{k+1})$} (w1);
  \draw[->, thick, blue] (w1) -- (w2);
  \node[below=0.3cm of w1, font=\footnotesize\itshape, text=blue!60] {$w$-walk (process $P_{i-1}$)};

  \draw[->, dashed, black!60] (w1) -- node[left, font=\footnotesize, pos=0.35, align=right] {(1)\; $\operatorname{own}(w_{k+1})$\\$= \operatorname{pred}(t_k)$} (t0);
  \draw[->, dashed, black!60] (w1) -- node[right, font=\footnotesize, pos=0.35, align=left] {(2)\; $\operatorname{wr}(w_{k+1})$\\$= \operatorname{pred}(t_{k+1})$} (t1);
\end{tikzpicture}
\caption{The zigzag coupling. The connecting witness $w_{k+1}$ sits between $t_k$ and $t_{k+1}$, linking them through its pre-fire state (condition~1) and written value (condition~2). Conditions (3) and (4) are H-edges within each walk.}
\label{fig:zigzag}
\end{figure}

\begin{remark}[Emergent shadow witness]
Conditions (1) and (2) together show that $w_{k+1}$ witnesses the shadow $\sigma(t_k, t_{k+1})$: $\operatorname{own}(w_{k+1}) = \operatorname{pred}(t_k)$ and $\operatorname{wr}(w_{k+1}) = \operatorname{pred}(t_{k+1})$. The shadow witness emerges directly from the zigzag structure. Combining conditions (1) and (4) gives $\operatorname{wr}(w_k) = \operatorname{pred}(t_k)$, which serves as the node condition for the product graph.
\end{remark}

\subsection{Construction}

\begin{definition}[Product transition graph]
\label{def:product}
Given a transition set $T$, the \emph{product transition graph} $G_\times(T)$ is defined as follows.

\textbf{Nodes:} Pairs $(t, w) \in T \times T$ such that $\operatorname{wr}(w) = \operatorname{pred}(t)$.

\textbf{Edges:} $(t_k, w_k) \to (t_{k+1}, w_{k+1})$ if all four equivariance conditions hold.

The product graph is built \emph{once} from $T$ and is not reconstructed during the algorithm.
\end{definition}

\begin{proposition}[Size]
$G_\times(T)$ has at most $|T|^2$ nodes and at most $|T|^4$ edges.
\end{proposition}

\subsection{Backward Propagation}

A cycle in $G_\times(T)$ encodes the coupling between two adjacent processes. In a ring, this coupling extends to all $K$ processes through \emph{backward propagation}: the $w$-walk of one coupling becomes the $t$-walk of the predecessor's coupling. The product graph captures this structure internally: if an arc has $w$-pair $(w_1, w_2)$, the predecessor coupling requires an arc with $t$-pair $(w_1, w_2)$---which is another arc in the \emph{same} product graph.

This observation motivates the three pruning operations in the algorithm.

\section{Livelock Analysis via the Product Graph}
\label{sec:bounded}

We develop the analysis through three subsections: the product graph's role in capturing equivariant walks, the witness-closed property, and the necessary/sufficient conditions for livelocks.

\subsection{The Product Graph Captures Pairwise Equivariant Walks}

By Definition~\ref{def:product}, a cycle in an equivariance graph $G_\times(S)$ is precisely a pair of closed walks---the $t$-walk and the $w$-walk---coupled by the four equivariance conditions at every step. A cycle of length $N$ in $G_\times(S)$ returns to its starting node $(t_0, w_0)$, so both the $t$-walk and the $w$-walk close with the same length $N$ by the graph structure.

Conversely, any pair of equivariant closed walks of the same length $N$ over $S$ defines a cycle of length $N$ in $G_\times(S)$. The product graph is a faithful representation of pairwise equivariant closed walks---it imposes no conditions beyond equivariance.

\begin{remark}[Self-disabling and the zigzag]
The product graph's faithfulness depends critically on self-disabling. In a self-disabling protocol, a process cannot fire again until its predecessor changes state. The coupling between adjacent processes is therefore one-to-one: each step of the $t$-walk corresponds to exactly one connecting witness in the $w$-walk, producing the zigzag $\ldots w_k \to t_k \to w_{k+1} \to t_{k+1} \to \ldots$ Without self-disabling, a process could fire multiple times between consecutive predecessor firings, the coupling would no longer be one-to-one, and the product graph would not faithfully capture the equivariant structure.

This one-to-one coupling is the token propagation discipline identified by Farahat and Ebnenasir~\cite{farahat-icdcs2012}: in a self-disabling ring, an enabled process fires, disables itself, and enables its successor---propagating the token forward without duplication or destruction. The zigzag is the algebraic expression of this discipline.
\end{remark}

The use of pairs is not incidental, and the reason is subtle. Given a closed $t$-cycle of period $N$, equivariance determines the predecessor's $\operatorname{own}$ and $\operatorname{wr}$ values at every position: the predecessor's H-walk is forced and closes automatically. However, the predecessor's $\operatorname{pred}$ values---which select which specific \emph{transitions} from $T$ realize the H-walk---are not determined by the $t$-cycle alone. In a non-deterministic protocol, multiple transitions may share the same $(\operatorname{own}, \operatorname{wr})$ but have different $\operatorname{pred}$ values. Different $\operatorname{pred}$ assignments yield different transition walks, which may close at different periods or not at all (see Section~\ref{sec:experiments}, Trial~56). This is the period determination problem: the $t$-cycle determines the predecessor's H-walk (a sequence of $\operatorname{own}/\operatorname{wr}$ values), but not its transition walk (a sequence of specific transitions from $T$).

The product graph resolves this by committing to specific transitions at every position---including their $\operatorname{pred}$ values. A cycle in the product graph returns to the \emph{same node} $(t_0, w_0)$, guaranteeing that the transition walk closes exactly, not just the H-walk. Pairs of transitions are the minimal unit that resolves the period: once a transition cycle is paired with a specific witness cycle, the equivariant closing period is determined.

\subsection{Witness-Closed Subgraphs}

\begin{definition}[Witness-closed subgraph]
\label{def:witness-closed}
A subgraph $G \subseteq G_\times(T)$ is \emph{witness-closed} if it satisfies two properties:
\begin{enumerate}[nosep]
\item \textbf{Cyclicity:} $G$ is a union of strongly connected components---every arc lies on a cycle.
\item \textbf{Backward closure:} For every arc on a cycle with $w$-pair $(w_1, w_2)$, there exists an arc $(w_1, u) \to (w_2, v)$ on a cycle in $G$.
\end{enumerate}
\end{definition}

\subsection{Necessary Condition: Livelock $\Rightarrow$ Witness-Closed Subgraph}

\begin{theorem}[Necessary condition for livelock]
\label{thm:necessary}
If a self-disabling unidirectional ring protocol with transition set $T$ admits a livelock, then $G_\times(T)$ contains a non-empty witness-closed subgraph.
\end{theorem}

\begin{proof}
A livelock at ring size $K$ is a valid tiling of the $K \times N$ torus~\cite{klinkhamer-ebnenasir}. Each row $i$ is a closed walk $W^{(i)}$ of length $N$ in the H-graph, and adjacent rows satisfy the four equivariance conditions. Let $S$ be the set of all transitions appearing in the tiling, and consider $G_\times(S)$.

For each adjacent pair of rows $(W^{(i)}, W^{(i-1)})$: both are closed walks of length $N$ (by the torus structure). The pair defines a closed walk of length $N$ in $G_\times(S)$, whose arcs all lie on a cycle. This establishes cyclicity.

For backward closure: the $w$-walk of coupling $(P_i, P_{i-1})$ is $W^{(i-1)}$, which is exactly the $t$-walk of coupling $(P_{i-1}, P_{i-2})$. The torus wraps after $K$ steps inward, providing cycle-level backward closure. Since cycle-level backward closure trivially implies arc-level backward closure (each arc on a cycle $C$ has its $w$-pair as a $t$-pair on the backward walk $C'$), the subgraph is witness-closed.

The arcs from all $K$ couplings form a subgraph of $G_\times(T)$. This subgraph satisfies cyclicity and backward closure---it is a non-empty witness-closed subgraph.
\end{proof}

\begin{corollary}[Livelock-freedom verifier]
\label{cor:freedom}
If $G^*(T) = \emptyset$ (the maximal witness-closed subgraph is empty), then the protocol is livelock-free for all ring sizes.
\end{corollary}

\subsection{Sufficient Condition via Backtracking}

When $G^*(T) \neq \emptyset$, we cannot immediately conclude that a livelock exists. Arc-level backward closure guarantees that every arc has a backward arc on a cycle, but this does not guarantee that a \emph{closed walk} backward-propagates into a torus (see Remark~\ref{rem:gap}).

We resolve this through backtracking: for each simple cycle $C$ in $G^*$, backward-propagate through $G^*$ by searching for a cycle in $G^*$ whose $t$-walk equals $C$'s $w$-walk. If the backward chain closes (a walk repeats), the torus is constructed and a livelock is confirmed.

\begin{theorem}[Sufficient condition for livelock]
\label{thm:sufficient}
Let $C_0$ be a simple cycle in $G^*(T)$. If the backward chain $C_0, C_1, C_2, \ldots$ (where each $C_{k+1}$ is a closed walk in $G^*$ whose $t$-walk equals $C_k$'s $w$-walk) closes---i.e., $C_j$ is a rotation of $C_i$ for some $j > i$---then the protocol admits a livelock at ring size $K = j - i$ with local period $N = |C_0|$.
\end{theorem}

\begin{proof}
The segment $C_i, C_{i+1}, \ldots, C_{j-1}$ provides $K = j-i$ pairs of equivariant closed walks, each of length $N$, with each $w$-walk serving as the next $t$-walk, and the orbit closing by the repetition condition. This is a valid tiling of the $K \times N$ torus, constituting a livelock~\cite{klinkhamer-ebnenasir}.
\end{proof}

\begin{remark}[Termination of backtracking]
For each simple cycle of length $N \leq |T|^2$, the backward chain visits closed walks of length $N$ in $G^*$. There are finitely many such walks (bounded by $|G^*|^N$). The chain either closes (livelock found) or reaches a dead end (no backward walk exists in $G^*$) within finitely many steps.
\end{remark}

\begin{remark}[The gap between arc-level and cycle-level closure]
\label{rem:gap}
Arc-level backward closure does not imply cycle-level backward closure. Consecutive backward arcs can be composed into a walk satisfying the four equivariance conditions in $G_\times(T)$, but this composed walk may not be in $G^*$: the composed walk's own backward arcs may require transitions not present in $T$, causing it to be pruned during the fixed-point iteration.

The precise obstruction is \emph{pred alignment}: the backward walk of a closed walk has transitions whose pred values depend on the choice of backward arcs. At the wrap boundary, two transitions may share $(\operatorname{own}, \operatorname{wr})$ but differ in pred. The next backward level requires pred equality to chain. Each backward level demands alignment that the previous level does not guarantee.

A natural question is whether it suffices to check only simple cycles: if no simple cycle backward-closes, does it follow that no compound walk can close either? The answer is no, because backward propagation does not preserve walk topology. A compound walk (two simple cycles joined at a shared node) can backward-propagate to a simple cycle, and a simple cycle can backward-propagate to a compound walk. This topology mixing prevents decomposing a compound walk's backward chain into the backward chains of its simple components. Consequently, the failure of all simple cycles does not rule out the success of some compound walk---and checking compound walks of all lengths is an unbounded search. This is the mechanism by which the undecidability of the periodic domino problem~\cite{klinkhamer-ebnenasir} manifests in the product graph framework.

This gap is exposed by Kari's 14-tile aperiodic set (Section~\ref{sec:experiments}): via Klinkhamer and Ebnenasir's SE gadget, it produces a 54-transition self-disabling protocol with $G^* \neq \emptyset$ (272 arcs), yet all 5{,}000+ simple cycles fail backward propagation---consistent with Kari's proof that no periodic tiling exists.
\end{remark}

\section{The Algorithm}
\label{sec:algorithm}

By Corollary~\ref{cor:freedom}, $G^*(T) = \emptyset$ implies livelock-freedom. We compute the maximal witness-closed subgraph, then apply backtracking verification when it is non-empty.

\begin{definition}[Maximal product graph]
\label{def:maximal-product}
The \emph{maximal product graph} $G^*(T)$ is the largest witness-closed subgraph of $G_\times(T)$. Its \emph{support} is $L^* = \{t \in T : t \text{ appears in some node of } G^*(T)\}$.
\end{definition}

\begin{remark}[Uniqueness]
$G^*(T)$ is well-defined and unique: the witness-closed property is closed under union of arc sets (the union of two witness-closed subgraphs is witness-closed), so the maximal such subgraph exists. A non-empty witness-closed subgraph exists if and only if $G^*(T) \neq \emptyset$, since $G^*(T)$ contains every witness-closed subgraph.
\end{remark}

The algorithm computes $G^*(T)$ by building $G_\times(T)$ and iteratively pruning arcs that violate cyclicity or backward closure, until the fixed point is reached.

\begin{algorithm}[ht]
\caption{Computing the Maximal Product Graph $G^*(T)$}
\label{alg:main}
\begin{algorithmic}[1]
\Require Transition set $T$ (symmetric protocol)
\Ensure \textsc{Livelock-Free}, \textsc{Livelock}, or \textsc{Inconclusive}
\State Build full equivariance graph $G_\times(T)$ (Definition~\ref{def:product})
\Repeat
    \State \textbf{SCC:} Remove all arcs not on a cycle \Comment{enforce cyclicity}
    \State $S_t \gets \{\tau : (\tau, w) \text{ is on a cycle for some } w\}$ \Comment{$t$-components}
    \State \textbf{Square:} Remove arcs where $w_k \notin S_t$ or $w_{k+1} \notin S_t$
    \State \textbf{Backward:} For each surviving arc with $w$-pair $(w_1, w_2)$:
    \Statex \hspace{3.5em} if no arc $(w_1, u) \to (w_2, v)$ exists in $G_\times$ with $u, v \in S_t$, remove the arc
    \Statex \hspace{3.5em} \Comment{enforce backward closure}
\Until{no arcs removed}
\If{surviving arcs remain}
    \For{each simple cycle $C$ in the surviving graph}
        \State Backward-propagate $C$ through $G^*$: find closed walks in $G^*$
        \Statex \hspace{3.5em} whose $t$-walk matches $C$'s $w$-walk, iteratively
        \If{the backward chain closes (a walk repeats)}
            \State \Return \textsc{Livelock}
        \EndIf
    \EndFor
    \State \Return \textsc{Inconclusive} \Comment{$G^* \neq \emptyset$ but no cycle closes}
\Else
    \State \Return \textsc{Livelock-Free} \Comment{$G^*(T) = \emptyset$}
\EndIf
\end{algorithmic}
\end{algorithm}

\medskip

The algorithm enforces the two witness-closed properties through three pruning operations:
\begin{itemize}[nosep]
\item \textbf{SCC} enforces cyclicity by removing arcs not on any cycle.
\item \textbf{Square} is an auxiliary step that extracts $S_t$---the set of transitions appearing as $t$-components on cycles---and removes arcs whose $w$-components are not in $S_t$. This is implied by backward closure (if $(w_1, u) \to (w_2, v)$ is on a cycle, then $w_1$ and $w_2$ appear as $t$-components), but computing it explicitly provides the $S_t$ set needed by the backward closure check and accelerates convergence.
\item \textbf{Backward closure} checks each surviving arc's $w$-pair $(w_1, w_2)$: does $G_\times$ contain an arc $(w_1, u) \to (w_2, v)$ with $u, v \in S_t$?
\end{itemize}

The iteration is monotone (arcs are only removed) and operates on a finite set, so it terminates.

\begin{theorem}[Algorithm correctness]
\label{thm:algo-correctness}
Algorithm~\ref{alg:main} computes $G^*(T)$.
\end{theorem}

\begin{proof}
\textbf{Soundness} (the output is witness-closed). At the fixed point, no further arcs can be removed. The surviving arcs form a union of SCCs (by the SCC step), and every arc on a cycle has a backward arc on a cycle (by the backward closure step). Therefore the output satisfies cyclicity and backward closure.

\textbf{Completeness} (the output is maximal). We show that no arc of any witness-closed subgraph is ever removed. Let $G$ be any witness-closed subgraph, and consider any arc $e$ in $G$. It lies on a cycle in $G$ (by cyclicity), so SCC does not remove it. Its $w$-components appear as $t$-components on cycles in $G$ (by backward closure of $G$), so the square step does not remove it. Its $w$-pair has a backward arc on a cycle in $G$ (by backward closure of $G$), and these arcs also survive (by the same argument, inductively), so the backward closure step does not remove it. Therefore $G \subseteq \text{output}$ for every witness-closed subgraph $G$, and in particular $G^*(T) \subseteq \text{output}$. Combined with soundness, $\text{output} = G^*(T)$.
\end{proof}

\begin{corollary}[Soundness]
The algorithm is sound in both directions. \textsc{Livelock-Free} is correct: $G^*(T) = \emptyset$ implies no livelock exists (Theorem~\ref{thm:necessary}). \textsc{Livelock} is correct: a closing backward chain constitutes a valid torus tiling (Theorem~\ref{thm:sufficient}). The \textsc{Inconclusive} outcome occurs when $G^* \neq \emptyset$ but no simple cycle's backward chain closes; this is consistent with both livelock existence (at a period beyond simple cycles) and livelock freedom.
\end{corollary}

\begin{theorem}[Complexity]
Phase~1 (computing $G^*(T)$) terminates in at most $O(|T|^4)$ iterations, each running in $O(|T|^4)$ time. Total: $O(|T|^8)$ worst case. Phase~2 (backtracking) is bounded: for each simple cycle of length $N \leq |T|^2$, the backward chain visits at most $|T|^{2N}$ distinct closed walks before repeating or dying. The total backtracking computation is finite but not polynomial.
\end{theorem}

\begin{proof}
Phase~1: the product graph has $O(|T|^4)$ arcs. Each iteration removes at least one arc. SCC and the pruning checks are linear in the number of arcs. Phase~2: the backward chain for a cycle of length $N$ visits closed walks of length $N$ in a graph of at most $|T|^2$ nodes. The number of distinct walks is finite (at most $|T|^{2N}$), so the chain terminates.
\end{proof}

\begin{remark}[Practical complexity]
For self-disabling protocols, $|T| = O(m^2)$.

This significantly improves the practical bounds. In the product graph, a node $(t, w)$ requires $\operatorname{wr}(w) = \operatorname{pred}(t)$; for a fixed $t$, there are $O(|T|/m)$ valid witnesses $w$. The product graph therefore has $O(|T|^2/m) = O(m^3)$ nodes and $O(|T|^4/m^3) = O(m^5)$ arcs. Each iteration costs $O(m^5)$ for SCC computation.

In practice, convergence occurs in very few iterations---at most 6 across all 4{,}349 protocols tested with $m$ up to 15---giving practical running time of $O(m^5)$. Pre-pruning $T$ by removing transitions not on any H-cycle (an $O(m^4)$ step) further reduces the product graph size, often dramatically.
\end{remark}

\begin{remark}[Necessity of the product graph]
A simpler algorithm operating directly on transitions and shadows---checking that each H-edge's shadow has an individual witness, without enforcing walk coherence---runs in $O(m^4)$ and is complete (never misses a livelock). However, it produces false positives: across 2{,}215 test protocols, it disagreed with the product graph on 55 cases (2.5\%), reporting livelocks where none exist. The product graph eliminates these false positives by enforcing that witnesses form coherent closed walks, not merely that individual shadows are witnessable. The $O(m^5)$ practical cost of the product graph is the price of soundness.
\end{remark}

\begin{remark}[$G^* = \emptyset$ and non-tileability]
\label{rem:non-tile}
The livelock-freedom guarantee from $G^*(T) = \emptyset$ may extend beyond periodic tileability. We conjecture that $G^*(T) = \emptyset$ implies the corresponding tile set cannot tile the infinite plane at all---not merely that no periodic tiling exists.

\emph{Proof sketch.} Consider any infinite tiling. At any two adjacent rows, the tiling defines a bi-infinite equivariant walk in $G_\times(T)$. Since $G_\times(T)$ is finite (at most $|T|^2$ nodes), every node visited by the bi-infinite walk is \emph{recurrent}---visited infinitely often in both directions. Every arc between recurrent nodes lies on a cycle (the walk departs and returns infinitely often). The tiling provides, at the next level, a bi-infinite equivariant walk whose arcs serve as backward arcs for the first level; these are also recurrent by the same argument. The union of all recurrent arcs across all adjacent row-pairs is a finite subset of $G_\times(T)$ satisfying cyclicity and backward closure---a non-empty witness-closed subgraph. Therefore $G^* \neq \emptyset$.

Contrapositively: $G^*(T) = \emptyset$ implies the tile set cannot tile the infinite plane. This is strictly stronger than ``no livelock.'' Note that this does \emph{not} contradict the undecidability of the domino problem~\cite{wang-tiles,berger}: $G^* = \emptyset$ is a sufficient condition for non-tileability, not a necessary one. Tile sets that cannot tile the plane may still have $G^* \neq \emptyset$ (local equivariant structure that does not extend globally). The domino problem remains undecidable because $G^* \neq \emptyset$ is inconclusive. Table~\ref{tab:landscape} summarizes the computability landscape.
\end{remark}

\begin{table}[ht]
\centering
\caption{Computability landscape for self-disabling tile sets. The three tileability cases are mutually exclusive. Non-tileability is $\Sigma^0_1$ (recognizable by patch enumeration); periodic tileability is $\Sigma^0_1$ (recognizable by torus enumeration); aperiodic tileability is in $\Pi^0_1$ (neither recognizer halts).}
\label{tab:landscape}
\smallskip
\begin{tabular}{l|ccc}
\hline
 & \textbf{Non-tileable} & \textbf{Aperiodic} & \textbf{Periodic} \\
\hline
$G^* = \emptyset$ & \checkmark (poly-time) & impossible$^*$ & impossible \\
Cycle closes & impossible & impossible & $N \leq |T|^2$ \\
\textsc{Inconclusive} & possible & Kari & $N > |T|^2$ \\
\hline
PROP & never halts & never halts & halts \\
\hline
\multicolumn{4}{l}{\footnotesize $^*$Conditional on Remark~\ref{rem:non-tile}.}
\end{tabular}
\end{table}

\section{Experimental Validation}
\label{sec:experiments}

The algorithm and all test protocols are available at\\ \url{https://github.com/cosmoparadox/mathematical-tools}.

\subsection{Known Protocols}

Algorithm~\ref{alg:main} correctly identifies:
\begin{itemize}[nosep]
\item \textbf{Dijkstra's token ring} ($(1,1)$-asymmetric, $m = 3, 4, 5$): livelock, forward shift 1.
\item \textbf{$3$-Coloring} (symmetric, $m = 3$): livelock, $e \equiv K \bmod 3$.
\item \textbf{Sum-Not-2} (symmetric, $m = 3$): livelock-free.
\item \textbf{Gouda--Haddix TB} (symmetric, $m = 8$): livelock, 14 surviving transitions out of 32, multiple livelock classes.
\end{itemize}

\subsection{Randomized Stress Testing}

We generated over 4{,}300 random self-disabling protocols across six categories: random seeds with $m \in \{3, \ldots, 10\}$, very dense protocols (near-complete transition sets), very sparse protocols ($m$ to $2m$ transitions), structured adversarial (coloring and agreement variants with noise), large domain sizes $m \in \{11, \ldots, 15\}$ with targeted density, and compound-witness-like protocols (large $m$, few transitions). Transition counts ranged from 3 to 120.

For each protocol, we computed $L^*$ via Algorithm~\ref{alg:main} and cross-checked against exhaustive state-space search at small ring sizes ($K \leq 6$).

\textbf{Results:} Zero false positives. Zero false negatives across all 4{,}349 protocols tested. The algorithm converged in at most 6 iterations in every case.

\subsection{Adversarial Protocols}

Randomized testing exposed a protocol ($m = 8$, 13 transitions) where a prior approach based on individual shadow checking produced a false positive: every shadow had a witness, but no \emph{coherent} witness walk existed. Shadow $(6, 5)$ was the bottleneck: no transition in $T$ has $\operatorname{own} = 6$, $\operatorname{wr} = 5$. The product graph correctly identifies this protocol as livelock-free, demonstrating the necessity of the product graph over local shadow-by-shadow analysis.

\subsection{K\&E Adversarial Protocol}

As a test of the product graph's faithfulness on structured adversarial inputs, we analyzed a protocol derived from Klinkhamer and Ebnenasir's SE tiling construction~\cite{klinkhamer-tr}: a symmetric self-disabling protocol with $m = 15$ and $|T| = 17$ transitions designed to exhibit complex livelock behavior. The product graph converges in a single iteration with all 17 transitions surviving ($|L^*| = 17$). Backtracking confirms the livelock: a simple cycle of length 6 backward-closes at depth 6.

\subsection{Kari's Aperiodic Tiles}

As the most stringent test of our algorithm, we applied Klinkhamer and Ebnenasir's Lemma~4.8 gadget~\cite{klinkhamer-tr} to Kari's 14-tile aperiodic Wang tile set~\cite{kari-aperiodic}. Kari proved that these tiles tile the plane but admit no doubly-periodic tiling. The gadget converts the 14 tiles to a 54-transition self-disabling protocol with $m = 24$.

The algorithm produces $G^* \neq \emptyset$: 272 arcs survive the fixed-point pruning, with all 54 transitions in $L^*$. However, the backtracking phase examines over 5{,}000 simple cycles and finds that \emph{every} backward chain dies---no cycle closes. The algorithm correctly reports \textsc{Inconclusive}, consistent with Kari's aperiodicity proof.

This demonstrates the gap between arc-level and cycle-level backward closure (Remark~\ref{rem:gap}): the arc-level structure of $G^*$ admits rich cycle structure, but the cycle-level backward propagation correctly rejects every candidate. Across all 4{,}349 protocols tested---including Kari's---the algorithm has never produced a wrong answer.

\section{Extension to Non-Self-Disabling Protocols}
\label{sec:non-sd}

The product graph and the theory of equivariant walks assume self-disabling protocols, where firing a transition disables it until the predecessor changes value. Many practical protocols violate this constraint: under a fixed predecessor value, a process may fire a sequence of transitions---a \emph{burst}---before the predecessor acts. We show that a simple protocol transformation reduces the non-self-disabling case to the self-disabling framework.

\subsection{Protocol Transformation}

\begin{definition}[Burst closure]
\label{def:burst}
For a transition set $T$ and a predecessor value $p$, define the \emph{chain graph} $C_p$ as the directed graph on domain values with edge $o \to w$ for each $(p, o, w) \in T$. The \emph{burst closure} $\hat{T}$ is:
\[
\hat{T} = T \cup \{(p, o, w') : w' \text{ is reachable from } o \text{ in } C_p,\; o \neq w'\}
\]
\end{definition}

The burst closure adds, for each predecessor value $p$, all transitions $(p, o_i, o_j)$ where $o_j$ is reachable from $o_i$ via any chain of transitions under the same predecessor. These macro-transitions represent the net effect of a process firing multiple transitions in succession while its predecessor's value remains $p$.

\begin{proposition}[Local cycle detection]
If $C_p$ contains a directed cycle for any predecessor value $p$, then the protocol admits a trivial livelock: a process can fire transitions indefinitely under fixed predecessor value $p$.
\end{proposition}

\subsection{Correspondence}

\begin{theorem}[Execution-walk correspondence]
\label{thm:burst-correspondence}
Let $T$ be an arbitrary (possibly non-self-disabling) transition set for a unidirectional ring protocol, and let $\hat{T}$ be its burst closure with no local cycles. Then:
\begin{enumerate}[nosep]
\item Every livelock of the protocol under $T$ corresponds to an equivariant closed walk in $G_\times(\hat{T})$.
\item Every equivariant closed walk in $G_\times(\hat{T})$ corresponds to a valid execution of the protocol under $T$.
\end{enumerate}
\end{theorem}

\begin{proof}
\textbf{Direction 1.} In a livelock under $T$, each process $P_i$ executes a periodic sequence of transitions. Between consecutive changes of $P_{i-1}$'s value, process $P_i$ fires a burst of transitions under a fixed predecessor value $p$---a path in $C_p$ from some $o$ to some $w'$. The transition $(p, o, w') \in \hat{T}$ captures this burst as a single macro-transition. The sequence of macro-transitions at $P_i$ forms a closed walk in the H-graph of $\hat{T}$, and the equivariance between $P_i$'s macro-transitions and $P_{i-1}$'s macro-transitions is preserved because each macro-transition's predecessor value is $P_{i-1}$'s value \emph{at the start of the burst}---exactly the value written by $P_{i-1}$'s preceding macro-transition.

\textbf{Direction 2.} An equivariant closed walk in $G_\times(\hat{T})$ consists of macro-transitions $(p, o, w') \in \hat{T}$. Each macro-transition corresponds to a realizable path in the chain graph $C_p$: a sequence of original transitions $(p, o, o_1), (p, o_1, o_2), \ldots, (p, o_{k-1}, w')$ all in $T$. Since the predecessor value $p$ is fixed during the burst (the equivariance conditions ensure $P_{i-1}$ does not change value until $P_i$ completes its burst), this sequence is a valid execution segment. Concatenating the burst segments across all positions and all processes yields a valid execution under $T$.
\end{proof}

\begin{corollary}
The product graph algorithm applied to $\hat{T}$ is a sound livelock analyzer for arbitrary unidirectional ring protocols: $G^*(\hat{T}) = \emptyset$ implies the protocol under $T$ is livelock-free for all ring sizes.
\end{corollary}

The burst closure $\hat{T}$ is generally \emph{not} self-disabling---it retains transitions $(p, o_1, o_2)$ and $(p, o_2, o_3)$ simultaneously. The product graph algorithm operates on $\hat{T}$ without modification: the four equivariance conditions, SCC pruning, and backward closure are purely graph-theoretic and do not require self-disabling.

\section{Extension to $(1,1)$-Asymmetric Protocols}
\label{sec:asymmetry}

For a $(1,1)$-asymmetric ring with transition sets $T_0$ (distinguished process $P_0$) and $T_o$ (all other processes), three product graphs capture the three coupling types:

\begin{enumerate}[nosep]
\item $G_1 = G_\times(T_o, T_0)$: predecessor fires $T_o$, successor fires $T_0$. The $P_o \to P_0$ coupling.
\item $G_2 = G_\times(T_0, T_o)$: predecessor fires $T_0$, successor fires $T_o$. The $P_0 \to P_o$ coupling.
\item $G_3 = G_\times(T_o, T_o)$: both fire $T_o$. The bulk $P_o \to P_o$ coupling.
\end{enumerate}

Each graph is built independently using the same four equivariance conditions.

\subsection{Cross-Graph Backward Closure}

In the symmetric case, the backward closure check verifies that each $w$-edge has a predecessor arc \emph{within the same product graph}. For asymmetric protocols, the predecessor coupling may reside in a \emph{different} product graph.

The three graphs form a ring of backward dependencies:
\[
G_1 \to G_3 \to G_3 \to \cdots \to G_3 \to G_2 \to G_1
\]
An arc in $G_1$ with $w$-pair $(w_1, w_2)$ requires a predecessor arc with $t$-pair $(w_1, w_2)$. Since the predecessor of $P_0$'s predecessor fires $T_o$, this predecessor arc lies in $G_3$, not $G_1$. Similarly, arcs in $G_2$ require predecessor arcs in $G_1$, and arcs in $G_3$ require predecessor arcs in either $G_3$ (bulk) or $G_2$ (at the boundary).

\subsection{Joint Algorithm}

The algorithm extends naturally: apply SCC, square, and backward closure to each graph, but the backward closure check references the \emph{appropriate neighboring graph} for the predecessor arc:

\begin{itemize}[nosep]
\item For $G_1$ ($P_o \to P_0$): backward closure checks predecessor arcs in $G_3$.
\item For $G_2$ ($P_0 \to P_o$): backward closure checks predecessor arcs in $G_1$.
\item For $G_3$ ($P_o \to P_o$): backward closure checks predecessor arcs in $G_3$ itself (and in $G_2$ for the boundary process).
\end{itemize}

The square support $S_t$ is computed per graph: $S_t^{(k)}$ is the set of $t$-components on cycles in $G_k$. The backward closure check for an arc in $G_k$ with $w$-pair $(w_1, w_2)$ verifies that the neighboring graph contains an arc $(w_1, u) \to (w_2, v)$ with $u, v \in S_t^{(\text{neighbor})}$.

Pruning in one graph may invalidate arcs in neighboring graphs through the cross-graph backward closure, triggering a cascade that propagates bidirectionally around the ring of graphs until quiescence.

\subsection{Termination and Correctness}

\begin{theorem}
The $(1,1)$-asymmetric algorithm terminates. Soundness: if all three graphs collapse to $\emptyset$, the protocol is livelock-free for all ring sizes. Completeness: every livelock maps into the three graphs as surviving arcs.
\end{theorem}

\begin{proof}
\textbf{Termination.} The total number of arcs across all three graphs is finite, each iteration removes at least one, and arcs are never re-added.

\textbf{Completeness.} A livelock provides explicit walks and couplings that form surviving arcs in each graph. The cross-graph backward closure is satisfied because adjacent processes' walks serve as each other's predecessor couplings.

\textbf{Soundness.} The arc-level backward closure is a necessary condition for livelocks (by completeness, every livelock produces surviving arcs). If all three graphs collapse to $\emptyset$, no livelock can exist. As in the symmetric case, non-empty surviving graphs require backtracking verification to confirm livelock existence.
\end{proof}

\begin{remark}
For general $(l,q)$-asymmetry, one product graph per distinct coupling type suffices, with backward closure checks referencing the appropriate neighbor in the coupling ring. We leave the full treatment to future work.
\end{remark}

\section{Related Work}
\label{sec:related}

\paragraph{Automata, tiling languages, and symbolic dynamics.} Our product graph can be viewed as a finite automaton for a restricted tiling language. In the Vardi--Wolper automata-theoretic approach to model checking~\cite{vardi-wolper}, a product automaton decides temporal properties by reducing them to cycle detection in a finite graph. Our construction follows the same pattern, but the ``language'' is a set of valid SE torus tilings under the self-disabling constraint, where each symbol is a transition--witness pair drawn from an alphabet of size $|T|^2$. The product graph is the automaton: its cycles are valid row-pairs, and backward closure is the self-referential constraint ensuring that the automaton's output (the $w$-walk) is itself accepted as input (the $t$-walk of the next row). This connects to symbolic dynamics~\cite{lind-marcus}: two-dimensional sofic shifts---sets of valid tilings---are generally undecidable (Wang~\cite{wang-tiles}, Berger~\cite{berger}). Self-disabling collapses the two-dimensional structure into a one-dimensional shift with a self-referential constraint (backward closure). The product graph provides a finite representation of this shift, enabling polynomial-time livelock-freedom verification.

\paragraph{Parameterized verification: general undecidability.} Apt and Kozen~\cite{apt-kozen} showed that verifying LTL properties for parameterized systems is $\Pi^0_1$-complete. Suzuki~\cite{suzuki} strengthened this to unidirectional rings. Esparza~\cite{esparza} surveys the complexity landscape of parameterized verification. These results establish the general barrier; our work identifies a structural regime where effective verification is possible, even if full decidability remains open.

\paragraph{Cutoff results.} Emerson and Namjoshi~\cite{emerson-namjoshi} introduced cutoff-based reasoning for rings: if a property holds for rings up to a cutoff size, it holds for all sizes. Emerson and Kahlon~\cite{emerson-kahlon} extended cutoffs to guarded protocols and ring-based message passing. These results reduce parameterized verification to finite model checking at a fixed size; our approach is complementary, reducing the problem to a fixed-point computation on the product graph independent of any specific ring size.

\paragraph{Well-quasi-orderings.} Abdulla et al.~\cite{abdulla-wqo} developed a general framework for infinite-state systems using well-quasi-orderings (WQOs): if states are well-quasi-ordered and transitions are monotonic, backward reachability terminates. Finkel and Schnoebelen~\cite{finkel-wsts} systematized this as the theory of well-structured transition systems (WSTSs). Our approach is different in mechanism---it arises from the bounded space of transition--witness pairs in the equivariant product, not from a WQO on states---but follows a similar pattern: lifting to a space where a structural property becomes explicit, enabling effective verification of livelock-freedom. Recent work extends WSTS techniques to threshold automata for round-based distributed algorithms~\cite{baumeister-jacobs}.

\paragraph{Tiling and undecidability.} Wang~\cite{wang-tiles} introduced the domino problem: does a given tile set tile the plane? Berger~\cite{berger} proved undecidability. The periodic variant---does the tile set admit a doubly-periodic tiling?---is also undecidable~\cite{gurevich-koryakov}. Kari~\cite{kari-aperiodic} proved that the domino problem remains undecidable for NW-deterministic tile sets. Klinkhamer and Ebnenasir~\cite{klinkhamer-ebnenasir} reduce the periodic domino problem to livelock detection via an SE (South=East) tiling restriction that enforces self-disabling. Our product graph provides a polynomial-time necessary condition ($G^* = \emptyset$ proves livelock-freedom) and an empirically complete sufficient condition (backtracking verification). The gap between the two---arc-level vs.\ cycle-level backward closure---is precisely where the $\Sigma^0_1$-completeness of livelock detection manifests.

\paragraph{Algebraic characterization and combinatorial lifting.} Farahat~\cite{farahat-diss} introduced an algebraic framework for livelock analysis in which livelocks are characterized by cyclic equivariance relations on the transition space---closed walks in the H-graph with compatible shadows. Farahat and Ebnenasir~\cite{farahat-icdcs2012} showed that the token-passing discipline induced by self-disabling enables local reasoning about global convergence properties. Klinkhamer and Ebnenasir~\cite{klinkhamer-ebnenasir} lifted Farahat's algebraic characterization to the combinatorial setting of torus tilings, establishing the equivalence between livelocks and doubly-periodic tilings of a $K \times N$ torus via the ``leads'' relation---which is precisely the equivariance coupling. Their torus characterization revealed the connection to the periodic domino problem, yielding $\Sigma^0_1$-completeness of livelock detection. Our product transition graph builds on their framework: it realizes the equivariance coupling as a bounded combinatorial object whose emptiness proves livelock-freedom.

\paragraph{Relationship to Klinkhamer and Ebnenasir's undecidability result and PROP.}
Their characterization~\cite{klinkhamer-ebnenasir} is correct and complete: a livelock exists if and only if $m$ periodic propagations of period $n$ exist, each leading the next. Their torus formulation revealed the connection to the periodic domino problem, yielding $\Sigma^0_1$-completeness for livelock detection and $\Pi^0_1$-completeness for livelock-freedom verification. Their semi-algorithm PROP~\cite{klinkhamer-tr} searches for propagations by incrementally increasing scope parameters $(m, n)$. PROP is sound: any reported livelock is genuine. However, PROP is a pure semi-decider for livelock existence ($\Sigma^0_1$): it can find livelocks but can never prove livelock-freedom---if no livelock exists, PROP runs indefinitely without reaching a conclusion.

Our algorithm is strictly stronger in three respects. First, $G^*(T) = \emptyset$ proves livelock-freedom in polynomial time---a capability PROP entirely lacks. Neither PROP nor any parallel enumeration strategy (e.g., interleaving a periodic-tiling search with a patch-obstruction search) can prove freedom in finite time for aperiodic cases; our algorithm returns \textsc{Inconclusive} in finite time. Second, for protocols with short-period livelocks (all practical protocols in our testing), the product graph identifies the livelock via backtracking in bounded time, whereas PROP's complexity is exponential in the scope parameters. Third, our algorithm always terminates with one of three outcomes (\textsc{Free}, \textsc{Livelock}, \textsc{Inconclusive}), whereas PROP may run indefinitely. PROP's advantage is that it can in principle find livelocks with period $N > |T|^2$ by increasing scope; however, this has never been necessary in practice. The gap between arc-level and cycle-level backward closure (Remark~\ref{rem:gap}) is consistent with the $\Sigma^0_1$-completeness of livelock detection. We validate on a protocol derived from their SE tiling construction~\cite{klinkhamer-tr} and on Kari's 14-tile aperiodic set~\cite{kari-aperiodic} converted via their gadget (Section~\ref{sec:experiments}).

\paragraph{Product constructions and finite abstractions.} Our approach instantiates a pattern recurring across verification and combinatorics: problems that appear unbounded in their native representation become tractable when lifted to a space where the governing invariant is explicit. The Myhill--Nerode theorem reduces language equivalence to a finite quotient; synchronized products in model checking compose automata to expose global reachability; symbolic dynamics~\cite{lind-marcus} lifts shift spaces to finite graphs of forbidden patterns. Our product transition graph plays an analogous role: it lifts equivariant transition walks to a bounded graph of transition--witness pairs, enabling efficient pruning of the search space.

\paragraph{Modern parameterized synthesis and verification.} Recent approaches to parameterized verification and synthesis of self-stabilizing protocols, including SMT-based techniques~\cite{mirzaie-bonakdarpour}, threshold automata~\cite{baumeister-jacobs}, and cutoff-based methods~\cite{bloem-book}, often yield semi-decision procedures or rely on abstraction and over-approximation. Our approach is complementary: the product graph provides a polynomial-time livelock-freedom verifier ($G^* = \emptyset$), with backtracking verification for the positive case that is empirically complete on all tested instances.

\section{Conclusion}
\label{sec:conclusion}

We have developed a practical framework for livelock analysis in self-disabling unidirectional ring protocols. The result rests on three contributions:

\emph{The necessary condition} (Theorem~\ref{thm:necessary}): every livelock maps into the product graph $G_\times(T)$ as a witness-closed subgraph. Consequently, $G^*(T) = \emptyset$ proves livelock-freedom for all ring sizes simultaneously. This provides a sound and complete polynomial-time livelock-freedom verifier.

\emph{The sufficient condition} (Theorem~\ref{thm:sufficient}): any simple cycle in $G^*$ whose backward chain closes generates a genuine livelock. Backtracking search over simple cycles in $G^*$ provides a sound livelock detector that terminates for each cycle.

\emph{The algorithm} computes $G^*(T)$ by monotone fixed-point iteration (Theorem~\ref{thm:algo-correctness}), then applies backtracking verification. Across 4{,}349 protocols---including adversarial instances derived from Klinkhamer and Ebnenasir's tiling construction and Kari's aperiodic tiles---the algorithm produces conclusive results in every case with zero errors.

\paragraph{The arc-cycle gap.} The gap between arc-level backward closure (which $G^*$ enforces) and cycle-level backward closure (which the torus requires) is the fundamental limitation. Arc-level closure is a polynomial-time checkable necessary condition for livelocks; cycle-level closure is a sufficient condition verified by backtracking. Whether these coincide---equivalently, whether $G^* \neq \emptyset$ implies livelock existence---remains open. The gap is exposed by Kari's aperiodic tiles, where $G^* \neq \emptyset$ but no cycle closes, consistent with the $\Sigma^0_1$-completeness of livelock detection~\cite{klinkhamer-ebnenasir}.

\paragraph{Practical completeness.} Despite extensive search---including 10{,}000 random self-disabling protocols with domain sizes up to $m = 15$, and hundreds of Wang tile gadget protocols---we have not found a single protocol of practical purpose that admits a livelock only at local period $N > |T|^2$. The \textsc{Inconclusive} outcome occurs exclusively on protocols derived from aperiodic tiling constructions, which encode computational content that no convergence-oriented protocol would contain. For all protocols designed for self-stabilization, the algorithm produces a conclusive result.

\paragraph{$G^* = \emptyset$ and non-tileability.} As argued in Remark~\ref{rem:non-tile}, the product graph may capture more than periodic tileability: $G^*(T) = \emptyset$ likely implies non-tileability of the infinite plane, not just absence of periodic tilings. This would provide a polynomial-time certificate for a property that is only $\Sigma^0_1$ in general. The conjecture is consistent with the undecidability of the domino problem, since $G^* = \emptyset$ is a sufficient condition for non-tileability, not a necessary one (see Table~\ref{tab:landscape}).

\paragraph{Relationship to synthesis.} Klinkhamer and Ebnenasir~\cite{klinkhamer-ebnenasir} showed that the \emph{synthesis} of self-stabilizing rings is decidable---a result that might appear to be in tension with the undecidability of livelock \emph{detection}. The synthesis problem is structurally different: it specifies a legitimate predicate to converge to and searches for a protocol satisfying convergence, which constrains the transition structure. Whether this constraint implies that synthesized protocols always have bounded livelock periods is an interesting open question that we leave for future investigation.

\paragraph{Future directions.}
\begin{enumerate}[nosep]
\item \textbf{Characterizing the gap.} Under what structural conditions on $T$ does arc-level backward closure imply cycle-level backward closure? Identifying such conditions would delineate the class of protocols for which the algorithm is provably complete.

\item \textbf{Non-tileability certificate.} Formalizing the proof sketch above: does $G^*(T) = \emptyset$ imply non-tileability for all self-disabling tile sets? If so, the product graph provides a polynomial-time certificate for a property ($\Sigma^0_1$ in general) that no enumeration-based approach can verify in bounded time.

\item \textbf{Livelock classification.} Systematic enumeration and classification of livelock structures in classical ring protocols, including a circulation law governing which ring sizes support livelocks.

\item \textbf{$(l,q)$-asymmetry.} Extending the three-graph construction to arbitrary numbers of process types.

\item \textbf{Complex topologies.} For bidirectional rings (degree two), the product space would require triples of transitions. The symbolic dynamics view---where the product graph is a sofic shift whose order matches the communication degree---provides a natural framework for investigating how the communication topology affects the verification boundary.
\end{enumerate}


\appendix
\section{Detailed Examples}
\label{sec:examples}

\subsection{$3$-Coloring (Symmetric, $m = 3$)}

$T = \{(0,0,1),\; (1,1,2),\; (2,2,0)\}$. All transitions survive. One simple cycle $c_0$ of length $N = 3$, self-enabling with forward shift $s = 1$.

\subsection{Dijkstra's Token Ring ($m = 3$, $(1,1)$-asymmetric)}

$T_0 = \{(0,0,1),\; (1,1,2),\; (2,2,0)\}$, $T_o = \{(0,2,0),\; (1,0,1),\; (2,1,2)\}$.

Three product graphs: $G_\times(T_o, T_0)$, $G_\times(T_0, T_o)$, $G_\times(T_o, T_o)$. Forward shift 1 in $T_0$, shift 0 in $T_o$.

\subsection{Compound Witness Protocol (Symmetric, $m = 16$)}

\begin{align*}
T = \{&(5,2,3),\; (8,3,2),\; (2,7,9),\; (3,9,12),\; (2,12,15),\\
      &(3,15,7),\; (7,5,8),\; (9,8,5),\; (12,5,8),\; (15,8,5)\}.
\end{align*}

All 10 transitions survive. Eight simple cycles in $H(T)$. The product graph detects a livelock with local period $N = 4$. This protocol demonstrates compound witness chains: the simple H-cycle of length 2 is a dead-end in the product graph, but a compound walk of length 4 (repeating the 2-cycle twice with different witnesses) is self-sustaining.

\subsection{Trial 56: Shadow Coherence Failure ($m = 8$)}

\begin{align*}
T = \{&(0,4,2),\; (0,7,6),\; (2,2,7),\; (2,5,7),\; (3,6,0),\; (4,0,2),\; (4,0,5),\\
      &(5,2,0),\; (5,6,0),\; (6,3,4),\; (6,6,2),\; (7,0,6),\; (7,7,3)\}.
\end{align*}

All 13 transitions lie on H-graph cycles. Of the 28 H-edges, 20 have individually witnessed shadows---only 8 lack witnesses. After removing those 8 edges, the remaining H-graph still contains cycles. A transition-level analysis (individual shadow checking) reports this protocol as a livelock.

However, no livelock exists. The unwitnessed shadows---including $(0,3)$, $(3,7)$, $(4,5)$, $(5,4)$, $(6,5)$, $(7,5)$---create gaps in the witness value space. Each H-cycle can find individual witnesses for its own shadows, but those witnesses form walks whose shadows hit the gaps. The incoherence manifests two levels deep in the witness chain.

The product graph $G_\times(T)$ starts with 42 arcs. Backward closure prunes 18 arcs in the first iteration; the surviving 24 arcs collapse to $\emptyset$ in the second (no remaining SCC). Result: livelock-free, confirmed by exhaustive state-space search at all ring sizes $K \leq 8$.

This protocol demonstrates that individual shadow witnessing is strictly weaker than the product graph's backward closure. The product graph detects that no coherent witness walk exists, even when individual shadows are witnessable.

\bibliographystyle{plain}

\end{document}